\documentclass[11pt]{article}
\usepackage[margin=1in]{geometry}
\usepackage{amsmath,amssymb,amsthm,mathtools}
\usepackage{enumitem,microtype}
\usepackage{authblk}
\usepackage[colorlinks=true,linkcolor=blue,citecolor=blue,urlcolor=blue]{hyperref}
\newtheorem{theorem}{Theorem}[section]
\newtheorem{proposition}[theorem]{Proposition}
\newtheorem{corollary}[theorem]{Corollary}

\theoremstyle{definition}
\newtheorem{definition}[theorem]{Definition}
\newtheorem{example}[theorem]{Example}
\theoremstyle{remark}
\newtheorem{remark}[theorem]{Remark}
\newcommand{\Tr}{\operatorname{Tr}}
\newcommand{\cS}{\mathcal{S}}
\newcommand{\cF}{\mathcal{F}}

\title{Trace-Norm Overlaps of Quantum States: Interpolation, Equality, and Data-Processing Rigidity}
\author[1]{Zahra Maleki Khouzani}
\author[1,2]{Seyed Mahmoud Manjegani\thanks{Corresponding author. Emails:
\href{mailto:Seyed.Manjegani@uregina.ca}{Seyed.Manjegani@uregina.ca};
\href{mailto:manjgani@iut.ac.ir}{manjgani@iut.ac.ir}.}}
\affil[1]{Department of Mathematical Sciences, Isfahan University of
Technology, Isfahan 84156-83111, Iran}
\affil[2]{Department of Mathematics and Statistics, University of Regina,
3737 Wascana Parkway, Regina, Saskatchewan S4S 0A2, Canada}
\date{}

\begin{document}
\maketitle

\begin{abstract}
Let $\rho$ and $\sigma$ be density operators on a separable Hilbert
space. For $0<\alpha<1$, we study the directed trace-norm overlap
\[
\Phi_\alpha(\rho,\sigma)
   =\|\rho^\alpha\sigma^{1-\alpha}\|_1
\]
and its symmetrized form
\[
\mathcal F_\alpha(\rho,\sigma)
   =\frac12\bigl(\Phi_\alpha(\rho,\sigma)
   +\Phi_{1-\alpha}(\rho,\sigma)\bigr).
\]
At $\alpha=\tfrac12$, both quantities coincide with the root
Uhlmann--Jozsa fidelity. Our aim is not to introduce a new notion of
fidelity, but to understand how these overlaps vary with the parameter
and when equality occurs in the resulting inequalities. We first prove
that $\alpha\mapsto\Phi_\alpha(\rho,\sigma)$ is log-convex. This yields
a sharp lower bound for $\mathcal F_\alpha$ in terms of the root
fidelity, together with a stronger intermediate geometric-mean bound.
We also show that $\Phi_\alpha$ is the trace functional
$Q_{\alpha,1/2}$ associated with the $\alpha$-$z$ R\'enyi divergence,
which connects the directed overlap with the known $\alpha$-$z$
R\'enyi theory.

We then determine the exact data-processing behavior of the
symmetrized family. Universal monotonicity under quantum channels holds
only at $\alpha=\tfrac12$; for every other value of $\alpha$, it already
fails under diagonal pinching of faithful real qubit states. In finite
dimensions, we give a complete spectral description of the equality
cases and provide explicit noncommuting examples. We also discuss some
natural questions about overlap-preserving maps. Finally, we prove
that the difference between the Petz overlap and the directed
trace-norm overlap vanishes exactly when $\rho$ and $\sigma$ commute,
without requiring either faithfulness or any additional support
assumption.
\end{abstract}



\section{Introduction}

For density operators $\rho$ and $\sigma$, the root Uhlmann--Jozsa
fidelity is
\[
 F(\rho,\sigma)
 =\|\rho^{1/2}\sigma^{1/2}\|_1
 =\Tr(\rho^{1/2}\sigma\rho^{1/2})^{1/2}.
\]
In some parts of the literature, the square of this quantity is also
called fidelity. Throughout this paper, however, we use the
root-fidelity convention.

The standard fidelity has several well-known properties. It can be
described by means of purifications~\cite{Uhlmann}. More precisely, a
mixed state can be obtained by reducing a pure state on a larger
Hilbert space, and the fidelity is the largest overlap between such
pure-state representations of the two states. The standard fidelity is
also monotone under quantum channels. It is related to the trace
distance by the Fuchs--van de Graaf inequalities~\cite{FuchsVanDeGraaf}
and gives rise to the Bures geometry.

The family studied here is not meant to replace the standard fidelity.
We wish to see whether the midpoint $\alpha=\tfrac12$ has a special
role among the trace-norm overlaps obtained by changing the powers of
the two states. Liang et al.~\cite{Liang} surveyed a related family of
norm-based quantities:

\[
 F_p^{\rm LB}(\rho,\sigma)
 =\frac{\|\rho^{1/2}\sigma^{1/2}\|_p^2}
 {\max\{\|\rho\|_p^2,\|\sigma\|_p^2\}},
 \qquad p\geq1.
\]
In this family the exponent $1/2$ is fixed and the Schatten index
varies. In our family the trace norm is fixed and the powers vary. Thus,
the two families vary different parameters. This leads us to Young and
H\H{o}lder inequalities, interpolation, and equality cases in
noncommutative $L^p$-spaces.

We first compare our quantity with the Petz overlap
\[
 Q_\alpha(\rho,\sigma)
 =\Tr(\rho^\alpha\sigma^{1-\alpha}),
\]
whose minimization gives the non-logarithmic quantum Chernoff quantity.
Since the product need not be positive,
\[
 Q_\alpha(\rho,\sigma)
 \leq\|\rho^\alpha\sigma^{1-\alpha}\|_1.
\]
Equality holds when the states commute. We prove that the converse is
also true, with no faithfulness or support assumptions.

The directed functional itself is already known in the theory of
quantum R\'enyi divergences. As we show in
Section~\ref{sec:alpha-z}, it is the $z=1/2$ trace functional in the
$\alpha$-$z$ family of Audenaert and Datta~\cite{AudenaertDatta}.
Log-convexity of related matrix norm functions is also known in the
general theory of unitarily invariant norms; see
Sababheh~\cite{Sababheh}. We therefore do not claim that
$\Phi_\alpha$ itself is new. Our main results concern its
symmetrization, its equality cases, and its behavior under quantum
channels.

Against this background, the main results of the paper may be
summarized as follows:
\begin{enumerate}[label=\textup{(\roman*)}]
\item log-convexity of $\alpha\mapsto\Phi_\alpha(\rho,\sigma)$;
\item a sharp comparison of $\cF_\alpha$ with the midpoint
      $F=\cF_{1/2}$;
\item complete finite-dimensional equality and constancy classifications,
      including noncommuting states and unequal overlapping supports;
\item an exact commutativity criterion for equality between the Petz and
      directed trace-norm overlaps;
\item the sharp universal data-processing range for the arithmetic
      symmetrization, together with faithful real-qubit pinching
      counterexamples away from the midpoint; and
\item comparison with the general von Neumann algebraic
      $\alpha$-$z$ theory.
\end{enumerate}

We now explain which part of the data-processing problem is not covered
by the known results. Since
\[
\Phi_\alpha(\rho,\sigma)=Q_{\alpha,1/2}(\rho,\sigma),
\]
the known results for the $\alpha$-$z$ R\'enyi functional determine the
range of parameters for which the directed overlap satisfies the
data-processing inequality (DPI). However, these results do not
determine whether the symmetrized quantity
\[
\mathcal F_\alpha(\rho,\sigma)
 =\frac12\bigl(\Phi_\alpha(\rho,\sigma)
 +\Phi_{1-\alpha}(\rho,\sigma)\bigr)
\]
has the same property. We prove that $\mathcal F_\alpha$ is monotone
under all quantum channels only when
$\alpha=\tfrac12$. To the best of our knowledge, this characterization
and the spectral condition for equality given below have not appeared
in the literature. Here the new results concern the symmetrized
functional and its equality cases. The functional $Q_{\alpha,1/2}$ and
its basic interpolation properties are already known.

The paper is organized as follows. We begin with some basic properties
of the directed and symmetrized overlaps, including the case of
commuting states. We then relate the directed overlap to the
$\alpha$-$z$ R\'enyi functional $Q_{\alpha,1/2}$ and prove its
log-convexity. This result is used to obtain sharp midpoint bounds and
to characterize the corresponding equality and constancy cases. We
next study the difference between the Petz overlap and the directed
trace-norm overlap. In the final main section, we determine exactly
when the symmetrized overlap is monotone under all quantum channels.
We end with some remarks on overlap-preserving maps and possible
extensions to semifinite von Neumann algebras.

\section{Definition and verified properties}

Let $H$ be a complex separable Hilbert space, and let $\Tr$ denote the
usual trace. We write $\mathcal S_1(H)$ for the trace class and set
\[
\mathcal S(H)
  =\{\rho\in\mathcal S_1(H):
      \rho\geq0,\ \Tr\rho=1\}.
\]      

\begin{definition}
For $\rho,\sigma\in\cS(H)$ and $0<\alpha<1$, define
\[
 \Phi_\alpha(\rho,\sigma)
 :=\|\rho^\alpha\sigma^{1-\alpha}\|_1
\]
and
\[
 \cF_\alpha(\rho,\sigma)
 :=\frac12\left(
 \Phi_\alpha(\rho,\sigma)
 +\Phi_{1-\alpha}(\rho,\sigma)\right).
\]
\end{definition}

\begin{proposition}\label{prop:basic-properties}
Let $\rho,\sigma\in\mathcal S(H)$ and $0<\alpha<1$. Then:
\begin{enumerate}[label=\textup{(\alph*)}]
\item
$
0\leq\Phi_\alpha(\rho,\sigma)\leq1
\quad\text{and}\quad
0\leq\mathcal F_\alpha(\rho,\sigma)\leq1;
$
\item
$
\mathcal F_\alpha(\rho,\sigma)
   =\mathcal F_\alpha(\sigma,\rho);
$
\item
$
\mathcal F_\alpha(\rho,\sigma)
   =\mathcal F_{1-\alpha}(\rho,\sigma);
$
\item
$
\mathcal F_{1/2}(\rho,\sigma)=F(\rho,\sigma);
$
\item $\mathcal F_\alpha$ is invariant under simultaneous unitary
conjugation;

\item
$
\mathcal F_\alpha(\rho,\sigma)=0
\quad\Longleftrightarrow\quad
\rho\sigma=0.
$

\end{enumerate}
\end{proposition}

\begin{proof}
Since $\rho$ and $\sigma$ are positive trace-class operators,
$\rho^\alpha\in\mathcal S_{1/\alpha}(H)$ and
$\sigma^{1-\alpha}\in\mathcal S_{1/(1-\alpha)}(H)$. The
noncommutative H\"older inequality therefore gives
\[
\begin{aligned}
\Phi_\alpha(\rho,\sigma)
 &=\|\rho^\alpha\sigma^{1-\alpha}\|_1\\
 &\leq
   \|\rho^\alpha\|_{1/\alpha}
   \|\sigma^{1-\alpha}\|_{1/(1-\alpha)}.
\end{aligned}
\]
Moreover,
\[
\|\rho^\alpha\|_{1/\alpha}
 =\bigl(\Tr(\rho^\alpha)^{1/\alpha}\bigr)^\alpha
 =(\Tr\rho)^\alpha=1,
\]
and similarly
\[
\|\sigma^{1-\alpha}\|_{1/(1-\alpha)}
 =(\Tr\sigma)^{1-\alpha}=1.
\]
Thus
\[
0\leq\Phi_\alpha(\rho,\sigma)\leq1.
\]
Applying the same argument with $\alpha$ replaced by $1-\alpha$ and
taking the arithmetic mean gives
\[
0\leq\mathcal F_\alpha(\rho,\sigma)\leq1.
\]
Next, using $\|x^*\|_1=\|x\|_1$, we obtain
\[
\begin{aligned}
\Phi_\alpha(\sigma,\rho)
 &=\|\sigma^\alpha\rho^{1-\alpha}\|_1\\
 &=\|(\sigma^\alpha\rho^{1-\alpha})^*\|_1\\
 &=\|\rho^{1-\alpha}\sigma^\alpha\|_1\\
 &=\Phi_{1-\alpha}(\rho,\sigma).
\end{aligned}
\]
Consequently,
\[
\begin{aligned}
\mathcal F_\alpha(\sigma,\rho)
 &=\frac12\bigl(
     \Phi_\alpha(\sigma,\rho)
     +\Phi_{1-\alpha}(\sigma,\rho)\bigr)\\
 &=\frac12\bigl(
     \Phi_{1-\alpha}(\rho,\sigma)
     +\Phi_\alpha(\rho,\sigma)\bigr)\\
 &=\mathcal F_\alpha(\rho,\sigma).
\end{aligned}
\]
This proves \textup{(b)}. Part \textup{(c)} follows directly from the
definition, since interchanging $\alpha$ and $1-\alpha$ merely
interchanges the two terms in the arithmetic mean.

At $\alpha=\tfrac12$, the two directed terms coincide, and hence
\[
\begin{aligned}
\mathcal F_{1/2}(\rho,\sigma)
 &=\|\rho^{1/2}\sigma^{1/2}\|_1\\
 &=F(\rho,\sigma),
\end{aligned}
\]
where $F$ denotes the root Uhlmann--Jozsa fidelity. This proves
\textup{(d)}.

Let $U$ be a unitary operator on $H$. By the functional calculus,
\[
(U\rho U^*)^\alpha=U\rho^\alpha U^*,
\qquad
(U\sigma U^*)^{1-\alpha}
   =U\sigma^{1-\alpha}U^*.
\]
Therefore,
\[
\begin{aligned}
\Phi_\alpha(U\rho U^*,U\sigma U^*)
 &=\|U\rho^\alpha\sigma^{1-\alpha}U^*\|_1\\
 &=\|\rho^\alpha\sigma^{1-\alpha}\|_1\\
 &=\Phi_\alpha(\rho,\sigma),
\end{aligned}
\]
where we used the unitary invariance of the trace norm. The same
identity holds with $\alpha$ replaced by $1-\alpha$, and thus proves
\textup{(e)}.

It remains to prove \textup{(f)}. Since both terms in the definition of
$\mathcal F_\alpha$ are nonnegative,
$
\mathcal F_\alpha(\rho,\sigma)=0
$
if and only if
$
\rho^\alpha\sigma^{1-\alpha}=0
\quad\text{and}\quad
\rho^{1-\alpha}\sigma^\alpha=0.
$
In fact, either one of these equalities is already sufficient to show
that the supports of $\rho$ and $\sigma$ are orthogonal.
To see this, suppose that
$
\rho^\alpha\sigma^{1-\alpha}=0.
$
Then
$
\operatorname{Ran}(\sigma^{1-\alpha})
   \subseteq\ker(\rho^\alpha).
$
Positive powers do not change the kernel or the support of a positive
operator. Hence
$
\ker(\rho^\alpha)=\ker\rho
$
and
$
\overline{\operatorname{Ran}(\sigma^{1-\alpha})}
   =s(\sigma)H,
$
where $s(\sigma)$ denotes the support projection of $\sigma$. Since
$\ker\rho=(I-s(\rho))H$ is closed, it follows that
$
s(\sigma)H\subseteq(I-s(\rho))H.
$
Therefore,
$
s(\rho)s(\sigma)=0.
$
As $\rho=\rho s(\rho)$ and $\sigma=s(\sigma)\sigma$, we conclude that
$
\rho\sigma
 =\rho\,s(\rho)s(\sigma)\,\sigma
 =0.
$

Conversely, suppose that $\rho\sigma=0$. Since $\rho$ and $\sigma$ are
positive, this is equivalent to the orthogonality of their support
projections:
$
s(\rho)s(\sigma)=0.
$
The operators $\rho^\alpha$ and $\rho$ have the same support, as do
$\sigma^{1-\alpha}$ and $\sigma$. Consequently,
\[
\rho^\alpha\sigma^{1-\alpha}=0
\quad\text{and}\quad
\rho^{1-\alpha}\sigma^\alpha=0.
\]
Thus both directed overlaps vanish, and therefore
$
\mathcal F_\alpha(\rho,\sigma)=0.
$
This completes the proof.
\end{proof}

\begin{proposition}
Let $\rho=|\xi\rangle\langle\xi|$ and
$\sigma=|\eta\rangle\langle\eta|$
be pure states, where $\xi,\eta\in H$ are unit vectors. Then, for every
$0<\alpha<1$,
\[
\mathcal{F}_{\alpha}(\rho,\sigma)
   =|\langle\xi,\eta\rangle|
   =F(\rho,\sigma).
\]
\end{proposition}

\begin{proof}
Since $\rho$ and $\sigma$ are rank-one projections, their
spectra are contained in $\{0,1\}$. It follows from the functional
calculus that
$\rho^t=\rho$ and $\sigma^t=\sigma$
for every $t>0$. Consequently,
\[
\rho^\alpha\sigma^{1-\alpha}=\rho\sigma, \quad \mbox{and}\quad
\rho\sigma =\langle\xi,\eta\rangle|\xi\rangle\langle\eta|.
\]
To compute its trace norm, set
$
T=|\xi\rangle\langle\eta|.
$
Then $T^*=|\eta\rangle\langle\xi|$ and hence
$
T^*T
 =\|\xi\|^2|\eta\rangle\langle\eta|.
$
The operator $T^*T$ vanishes on $\eta^\perp$, while
$
T^*T\eta
 =\|\xi\|^2\|\eta\|^2\eta.
$
Thus $T^*T$ has only one nonzero eigenvalue, namely
$\|\xi\|^2\|\eta\|^2$. Therefore, $T$ has only one nonzero
singular value, equal to
$
\|\xi\|\,\|\eta\|.
$
Since $\xi$ and $\eta$ are unit vectors, this singular value is equal
to $1$, and consequently
$
\bigl\||\xi\rangle\langle\eta|\bigr\|_1=1.
$
Therefore,
$
\|\rho\sigma\|_1
 =|\langle\xi,\eta\rangle|.
$
It follows that
\[
\Phi_\alpha(\rho,\sigma)
 =\|\rho^\alpha\sigma^{1-\alpha}\|_1
 =|\langle\xi,\eta\rangle|.
\]
Similarly,
\[
\Phi_{1-\alpha}(\rho,\sigma)
 =\|\rho^{1-\alpha}\sigma^\alpha\|_1
 =|\langle\xi,\eta\rangle|.
\]
Hence,
\[
\begin{aligned}
\mathcal{F}_\alpha(\rho,\sigma)
 &=\frac{1}{2}
   \left(
     \Phi_\alpha(\rho,\sigma)
     +\Phi_{1-\alpha}(\rho,\sigma)
   \right) \\
 &=|\langle\xi,\eta\rangle|.
\end{aligned}
\]
Finally, since $\rho^{1/2}=\rho$ and $\sigma^{1/2}=\sigma$, the root
Uhlmann--Jozsa fidelity satisfies
\[
F(\rho,\sigma)
 =\|\rho^{1/2}\sigma^{1/2}\|_1
 =\|\rho\sigma\|_1
 =|\langle\xi,\eta\rangle|.
\]
This completes the proof.
\end{proof}

\begin{remark}
If only $\rho=|\xi\rangle\langle\xi|$ is pure, then
\[
 \Phi_\alpha(\rho,\sigma)
 =\|\sigma^{1-\alpha}\xi\|,
 \qquad
 \Phi_{1-\alpha}(\rho,\sigma)
 =\|\sigma^\alpha\xi\|.
\]
These quantities generally depend on $\alpha$ and do not equal the usual
pure--mixed fidelity. Hence the functional should initially be called a
parameterized trace-norm overlap, not automatically a quantum fidelity
satisfying all of Jozsa's axioms.
\end{remark}


If $\rho$ and $\sigma$ commute, write
$
 \rho=\operatorname{diag}(p_j),\qquad
 \sigma=\operatorname{diag}(q_j).
$
Then
\[
 \Phi_\alpha=\sum_jp_j^\alpha q_j^{1-\alpha},
 \qquad
 \cF_\alpha=\frac12\sum_j
 \left(p_j^\alpha q_j^{1-\alpha}
 +p_j^{1-\alpha}q_j^\alpha\right).
\]

\begin{theorem}
If $\rho$ and $\sigma$ commute, then
\[
 F(\rho,\sigma)=\cF_{1/2}(\rho,\sigma)
 \leq\cF_\alpha(\rho,\sigma)\leq1.
\]
For $\alpha\neq\tfrac12$, equality on the left holds precisely when
$p_j=q_j$ for every $j$ in the common support.
\end{theorem}

\begin{proof}
Since $\rho$ and $\sigma$ are commuting positive compact operators,
they admit a common orthonormal eigenbasis. Thus, we may write
\[
\rho=\sum_j p_j |e_j\rangle\langle e_j|,
\qquad
\sigma=\sum_j q_j |e_j\rangle\langle e_j|,
\]
where
\[
p_j,q_j\geq0,
\qquad
\sum_jp_j=\sum_jq_j=1.
\]
In this representation, the common support corresponds to those
indices $j$ for which $p_jq_j>0$.
By the functional calculus,
\[
\rho^\alpha\sigma^{1-\alpha}
 =\sum_jp_j^\alpha q_j^{1-\alpha}
   |e_j\rangle\langle e_j|.
\]
Since this operator is positive, its trace norm is equal to its trace.
Therefore,
\[
\Phi_\alpha(\rho,\sigma)
 =\sum_jp_j^\alpha q_j^{1-\alpha}.
\]
Similarly,
\[
\Phi_{1-\alpha}(\rho,\sigma)
 =\sum_jp_j^{1-\alpha}q_j^\alpha.
\]
It follows that
\[
\mathcal{F}_\alpha(\rho,\sigma)
 =\frac12\sum_j
 \left(
 p_j^\alpha q_j^{1-\alpha}
 +p_j^{1-\alpha}q_j^\alpha
 \right).
\]

For each $j$, the arithmetic--geometric mean inequality gives
\[
\begin{aligned}
\frac12\left(
p_j^\alpha q_j^{1-\alpha}
+p_j^{1-\alpha}q_j^\alpha
\right)
&\geq
\sqrt{
p_j^\alpha q_j^{1-\alpha}
p_j^{1-\alpha}q_j^\alpha
}\\
&=\sqrt{p_jq_j}.
\end{aligned}
\]
Summing over $j$, we obtain
\[
\mathcal{F}_\alpha(\rho,\sigma)
 \geq\sum_j\sqrt{p_jq_j}.
\]
On the other hand, since $\rho$ and $\sigma$ commute,
\[
F(\rho,\sigma)
 =\|\rho^{1/2}\sigma^{1/2}\|_1
 =\sum_j\sqrt{p_jq_j}.
\]
Hence,
\[
F(\rho,\sigma)
 =\mathcal{F}_{1/2}(\rho,\sigma)
 \leq\mathcal{F}_\alpha(\rho,\sigma).
\]
The upper bound
$
\mathcal{F}_\alpha(\rho,\sigma)\leq1
$
follows directly from Proposition~\ref{prop:basic-properties}.

We finally consider equality in the lower bound. Since each term in
the preceding arithmetic--geometric mean inequality is nonnegative,
equality after summation holds if and only if equality holds for every
$j$. If $p_jq_j=0$, both sides of the corresponding inequality are
zero. If $p_jq_j>0$, equality holds if and only if
\[
p_j^\alpha q_j^{1-\alpha}
 =p_j^{1-\alpha}q_j^\alpha,
\]
which is equivalent to
\[
\left(\frac{p_j}{q_j}\right)^{2\alpha-1}=1.
\]
For $\alpha\neq\tfrac12$, this occurs precisely when
\[
p_j=q_j.
\]
Thus equality in the lower bound holds exactly when $p_j=q_j$ for
every $j$ in the common support.
\end{proof}

\section{Identification with the \texorpdfstring{$\alpha$-$z$}{alpha-z}
R\'enyi family}
\label{sec:alpha-z}

For positive matrices and parameters $\beta,z>0$, the trace functional
underlying the $\alpha$-$z$ R\'enyi divergence is
\[
 Q_{\beta,z}(\rho\Vert\sigma)
 :=\Tr\left(
 \sigma^{(1-\beta)/(2z)}
 \rho^{\beta/z}
 \sigma^{(1-\beta)/(2z)}
 \right)^z.
\]
For normalized states,
\[
 D_{\beta,z}(\rho\Vert\sigma)
 =\frac{1}{\beta-1}
   \log Q_{\beta,z}(\rho\Vert\sigma),
 \qquad \beta\ne1.
\]
This family was introduced and systematically studied by Audenaert and
Datta~\cite{AudenaertDatta}; its von Neumann algebraic theory has recently
been developed by Hiai and Jen\v{c}ov\'a~\cite{HiaiJencova}.

\begin{proposition}\label{prop:alpha-z}
For $0<\alpha<1$,
\[
 \Phi_\alpha(\rho,\sigma)
 =Q_{\alpha,1/2}(\rho\Vert\sigma),
\]
and therefore
\[
 D_{\alpha,1/2}(\rho\Vert\sigma)
 =\frac{1}{\alpha-1}\log\Phi_\alpha(\rho,\sigma).
\]
At $\alpha=1/2$, this reduces to
\[
 D_{1/2,1/2}(\rho\Vert\sigma)=-2\log F(\rho,
 \sigma).
\]
\end{proposition}

\begin{proof}
Since
\[
 |\rho^\alpha\sigma^{1-\alpha}|^2
 =\sigma^{1-\alpha}\rho^{2\alpha}
  \sigma^{1-\alpha},
\]
we have
\[
 \Phi_\alpha(\rho,\sigma)
 =\Tr\left(
 \sigma^{1-\alpha}\rho^{2\alpha}
 \sigma^{1-\alpha}
 \right)^{1/2},
\]
which is exactly $Q_{\alpha,1/2}(\rho\Vert\sigma)$.
The remaining assertions follow from the definition of
$D_{\alpha,z}$.
\end{proof}

\begin{remark}
By Proposition~\ref{prop:alpha-z}, the definition, tensor-product
multiplicativity, and data processing for the directed quantity belong
to the known $\alpha$-$z$ theory. Also, log-convexity of matrix norm
functions of the form $t\mapsto\|A^tXB^{1-t}\|$ was proved for
unitarily invariant norms in~\cite{Sababheh}. We include
Theorem~\ref{thm:logconvex} because its proof is short, applies directly
to density operators, and extends explicitly to separable Hilbert
spaces. Our new results concern the symmetrized quantity. They include
the sharp DPI range in Theorem~\ref{thm:dpi-counterexample} and the
finite-dimensional spectral-edge description in
Theorem~\ref{thm:spectral-equality}. To the best of our knowledge, these
two results do not appear in the existing $\alpha$-$z$ or
matrix-interpolation literature.
\end{remark}

\section{Log-convexity and the midpoint inequality}

The following interpolation theorem gives the main inequality for
noncommuting density operators. No faithfulness assumption is needed,
so the states may have nonzero kernels.

\begin{theorem}
\label{thm:logconvex}
Let $\rho,\sigma\in\mathcal{S}(H)$. Suppose that $0<a<b<1,\quad 0<\theta<1$,
and set $t=(1-\theta)a+\theta b$. Then
$
\Phi_t(\rho,\sigma)
\leq
\Phi_a(\rho,\sigma)^{1-\theta}
\Phi_b(\rho,\sigma)^\theta.
$
Consequently, the function $t\longmapsto\Phi_t(\rho,\sigma)$ is log-convex on $(0,1)$.
\end{theorem}

\begin{proof}
We first prove the result when $H$ is finite-dimensional. Put $d=b-a$
and consider the closed strip
\[
\mathcal{D}
 =\{z\in\mathbb{C}:0\leq\operatorname{Re}z\leq1\}.
\]
For $z\in\mathcal{D}$, define $X(z)=\rho^{a+dz}\sigma^{1-a-dz}$.
The real parts of the two exponents satisfy
\[
\operatorname{Re}(a+dz)\geq a>0,\quad\mbox{and}\quad
\operatorname{Re}(1-a-dz)\geq1-b>0.
\]
Thus no negative power of $\rho$ or $\sigma$ occurs. For a positive
operator $A$, we define
\[
A^w
 =\sum_{\lambda>0}\lambda^w P_\lambda,
\qquad
\lambda^w=e^{w\log\lambda},
\]
where the sum is taken over the positive eigenvalues of $A$. In
particular, $A^w$ is taken to be zero on $\ker A$. With this convention,
$X(z)$ is continuous on $\mathcal{D}$ and analytic in its interior.
Let $K$ be an operator satisfying $\|K\|_\infty\leq1$, and define
$
f_K(z)=\Tr\bigl(K^*X(z)\bigr).
$
Then $f_K$ is continuous on the closed strip and analytic in its
interior.

We first estimate $f_K$ on the boundary line
$\operatorname{Re}z=0$. If $y\in\mathbb{R}$, then
\[
\begin{aligned}
X(iy)
 &=\rho^{a+idy}\sigma^{1-a-idy}\\
 &=\rho^{idy}
   \bigl(\rho^a\sigma^{1-a}\bigr)
   \sigma^{-idy}.
\end{aligned}
\]
On the support of $\rho$, the operator $\rho^{idy}$ is unitary, while
it is zero on $\ker\rho$. Similarly, $\sigma^{-idy}$ is unitary on the
support of $\sigma$. Since
\[
\rho^a\sigma^{1-a}
 =s(\rho)\rho^a\sigma^{1-a}s(\sigma),
\]
multiplication by these support unitaries does not change its singular
values. Therefore,
\[
\|X(iy)\|_1
 =\|\rho^a\sigma^{1-a}\|_1
 =\Phi_a(\rho,\sigma).
\]
Using trace duality, we obtain
\[
\begin{aligned}
|f_K(iy)|
 &=\left|\Tr\bigl(K^*X(iy)\bigr)\right|\\
 &\leq\|K\|_\infty\|X(iy)\|_1\\
 &\leq\Phi_a(\rho,\sigma).
\end{aligned}
\]
The same argument on the other boundary line gives
\[
\|X(1+iy)\|_1
 =\|\rho^b\sigma^{1-b}\|_1
 =\Phi_b(\rho,\sigma),
\]
and hence
$
|f_K(1+iy)|
 \leq\Phi_b(\rho,\sigma)
$
for every $y\in\mathbb{R}$.
By the Hadamard three-lines theorem
(see \cite[Theorem~23.1.1]{BealsWong}), we obtain
\[
|f_K(\theta)|
\leq
\Phi_a(\rho,\sigma)^{1-\theta}
\Phi_b(\rho,\sigma)^\theta.
\]
Since
\[
a+d\theta
 =a+(b-a)\theta
 =(1-\theta)a+\theta b
 =t,
\]
we have
\[
X(\theta)=\rho^t\sigma^{1-t}.
\]
Therefore,
\[
\left|
\Tr\bigl(K^*\rho^t\sigma^{1-t}\bigr)
\right|
\leq
\Phi_a(\rho,\sigma)^{1-\theta}
\Phi_b(\rho,\sigma)^\theta.
\]
Taking the supremum over all $K$ with $\|K\|_\infty\leq1$ and using
the dual characterization of the trace norm,
\[
\|Y\|_1
 =\sup_{\|K\|_\infty\leq1}
   |\Tr(K^*Y)|,
\]
we obtain
\[
\Phi_t(\rho,\sigma)
\leq
\Phi_a(\rho,\sigma)^{1-\theta}
\Phi_b(\rho,\sigma)^\theta.
\]
This proves the result in finite dimensions.

We now consider a separable Hilbert space $H$. For each $n\geq1$, set
\[
\rho_n
 =\rho\,1_{[1/n,\infty)}(\rho),
\qquad
\sigma_n
 =\sigma\,1_{[1/n,\infty)}(\sigma).
\]
Since $\rho$ is a positive trace-class operator, its eigenvalues,
counted with multiplicity, have a finite sum. Hence only finitely many
of them can be greater than or equal to $1/n$. More precisely, if
\[
e_n=1_{[1/n,\infty)}(\rho),
\]
then
\[
\rho e_n\geq\frac{1}{n}e_n.
\]
Therefore,
\[
1=\Tr\rho
\geq\Tr(\rho e_n)
\geq\frac{1}{n}\Tr(e_n).
\]
It follows that
\[
\operatorname{rank}(e_n)=\Tr(e_n)\leq n,
\]
and hence $\rho_n=\rho e_n$ has finite rank. The same argument applies
to $\sigma_n$.

Although $\rho_n$ and $\sigma_n$ need not have trace one, the
finite-dimensional argument above applies to arbitrary positive
finite-rank operators. We apply it on the finite-dimensional subspace
\[
H_n
 =s(\rho_n)H+s(\sigma_n)H.
\]
We next show that the corresponding overlaps converge. Let $0<r<1$.
By the spectral theorem,
\[
\|\rho^r-\rho_n^r\|_{1/r}^{1/r}
 =\Tr\bigl(\rho\,1_{(0,1/n)}(\rho)\bigr).
\]
The right-hand side tends to zero as $n\to\infty$. Hence
$\|\rho^r-\rho_n^r\|_{1/r}\longrightarrow0.
$
In the same way,
$\|\sigma^r-\sigma_n^r\|_{1/r}\longrightarrow0$.
We write
\[
\rho_n^r\sigma_n^{1-r}
 -\rho^r\sigma^{1-r}=
 (\rho_n^r-\rho^r)\sigma_n^{1-r}+
 \rho^r(\sigma_n^{1-r}-\sigma^{1-r}).
\]
The noncommutative H\"older inequality therefore gives
\[
\|\rho_n^r\sigma_n^{1-r} -\rho^r\sigma^{1-r}\|_1
\leq\|\rho_n^r-\rho^r\|_{1/r}
\|\sigma_n^{1-r}\|_{1/(1-r)}+
\|\rho^r\|_{1/r}
\|\sigma_n^{1-r}-\sigma^{1-r}\|_{1/(1-r)}.
\]
Moreover, $\|\rho^r\|_{1/r}=1$
and $\|\sigma_n^{1-r}\|_{1/(1-r)}=(\Tr\sigma_n)^{1-r}\leq1$.
It follows that
\[
\|\rho_n^r\sigma_n^{1-r}
      -\rho^r\sigma^{1-r}\|_1
\longrightarrow0.
\]
Consequently,
$\Phi_r(\rho_n,\sigma_n)
\longrightarrow
\Phi_r(\rho,\sigma)
$
for every $r\in(0,1)$.

Applying the finite-dimensional inequality to $\rho_n$ and $\sigma_n$
gives
\[
\Phi_t(\rho_n,\sigma_n)
\leq
\Phi_a(\rho_n,\sigma_n)^{1-\theta}
\Phi_b(\rho_n,\sigma_n)^\theta.
\]
Letting $n\to\infty$, we obtain
\[
\Phi_t(\rho,\sigma)
\leq
\Phi_a(\rho,\sigma)^{1-\theta}
\Phi_b(\rho,\sigma)^\theta.
\]
This completes the proof.
\end{proof}

\begin{corollary}
\label{thm:midpoint}
For any $\rho,\sigma\in\mathcal{S}(H)$ and every $0<\alpha<1$, 
\[
 \|\rho^{1/2}\sigma^{1/2}\|_1
 \leq
 \sqrt{\|\rho^\alpha\sigma^{1-\alpha}\|_1
       \|\rho^{1-\alpha}\sigma^\alpha\|_1}
 \leq\frac12\left(
 \|\rho^\alpha\sigma^{1-\alpha}\|_1+
 \|\rho^{1-\alpha}\sigma^\alpha\|_1
 \right)\leq1.
\]
Equivalently,
\[
 F(\rho,\sigma)\leq
 \sqrt{\Phi_\alpha(\rho,\sigma)
       \Phi_{1-\alpha}(\rho,\sigma)}
 \leq\mathcal F_\alpha(\rho,\sigma)\leq1.
\]
These bounds are sharp.
\end{corollary}

\begin{proof}
Apply Theorem~\ref{thm:logconvex} at the midpoint of $\alpha$ and
$1-\alpha$ to obtain
\[
 \Phi_{1/2}\leq
 \sqrt{\Phi_\alpha\Phi_{1-\alpha}}.
\]
The second inequality follows from the usual arithmetic--geometric
mean inequality, and the final bound follows from the noncommutative
H\"older inequality. These bounds cannot be improved in general.
Indeed, if $\rho=\sigma$, then
\[
\Phi_\alpha(\rho,\rho)
=\Phi_{1-\alpha}(\rho,\rho)
=F(\rho,\rho)=1,
\]
and hence equality holds throughout the chain.
\end{proof}

\begin{theorem}
\label{thm:midpoint-equality}
Fix $0<\alpha<1$ with $\alpha\neq\tfrac12$. Then the equality cases in
Corollary~\ref{thm:midpoint} are described as follows.

\begin{enumerate}[label=\textnormal{(\roman*)}]
\item
For any $\rho,\sigma\in\mathcal{S}(H)$,
$
\sqrt{
\Phi_\alpha(\rho,\sigma)
\Phi_{1-\alpha}(\rho,\sigma)}=
\mathcal{F}_\alpha(\rho,\sigma)
$
if and only if
$
\Phi_\alpha(\rho,\sigma)=\Phi_{1-\alpha}(\rho,\sigma).
$

\item
Suppose that $H$ is finite-dimensional and that $\rho$ and $\sigma$
have the same support. Then
$
F(\rho,\sigma)^2=\Phi_\alpha(\rho,\sigma)\Phi_{1-\alpha}(\rho,\sigma)
$
if and only if
$\rho=\sigma$.

\item
Under the assumptions of \textnormal{(ii)},
$
F(\rho,\sigma)=\mathcal{F}_\alpha(\rho,\sigma)
$
if and only if $\rho=\sigma$.

\item
Suppose that $\rho$ and $\sigma$ commute, and write their eigenvalues
in a common orthonormal eigenbasis as $(p_i)$ and $(q_i)$, respectively.
If their supports are not orthogonal, then
\[
F(\rho,\sigma)^2
=
\Phi_\alpha(\rho,\sigma)
\Phi_{1-\alpha}(\rho,\sigma)
\]
if and only if there exists a constant $c>0$ such that
$p_i=cq_i$ whenever $p_iq_i>0$.

\item
Under the commuting assumption in \textnormal{(iv)},
$F(\rho,\sigma)=\mathcal{F}_\alpha(\rho,\sigma)$
if and only if either the supports of $\rho$ and $\sigma$ are
orthogonal, or $p_i=q_i$ whenever $p_iq_i>0$.

\end{enumerate}

When the supports of $\rho$ and $\sigma$ are orthogonal, we have
\[
F(\rho,\sigma)
=
\Phi_\alpha(\rho,\sigma)
=
\Phi_{1-\alpha}(\rho,\sigma)
=
\mathcal{F}_\alpha(\rho,\sigma)
=0.
\]
Thus, in this case, equality holds in the first two inequalities of
the midpoint chain at the value zero.
\end{theorem}

\begin{proof}
Part \textnormal{(i)} is precisely the equality condition in the scalar
arithmetic--geometric mean inequality.
We now prove \textnormal{(ii)}. Since $\rho$ and $\sigma$ have the same
support, we may restrict them to the subspace
\[
s(\rho)H=s(\sigma)H.
\]
On this subspace, both operators are positive definite. This
restriction does not change any of the quantities appearing in the
equality. Define
\[
g(t)=\log\Phi_t(\rho,\sigma),
\qquad 0<t<1.
\]
By Theorem~\ref{thm:logconvex}, the function $g$ is convex on $(0,1)$.
Let
\[
a=\min\{\alpha,1-\alpha\},
\qquad
b=\max\{\alpha,1-\alpha\}.
\]
Since $\alpha\neq\tfrac12$, we have
\[
0<a<\tfrac12<b<1
\qquad\text{and}\qquad
\frac{a+b}{2}=\frac12.
\]
The equality
\[
F(\rho,\sigma)^2
=
\Phi_\alpha(\rho,\sigma)
\Phi_{1-\alpha}(\rho,\sigma)
\]
can be written as
\[
2g\left(\frac12\right)=g(a)+g(b),
\]
because
\[
F(\rho,\sigma)=\Phi_{1/2}(\rho,\sigma).
\]
Thus equality holds in the convexity inequality at the midpoint of
$a$ and $b$. A convex function can have equality at an interior point
of a chord only when it agrees with that chord on the whole interval.
Therefore, $g$ is affine on $[a,b]$.

We next show that this affine behavior extends to all of $(0,1)$.
Since $\rho$ and $\sigma$ are positive definite matrices, the maps
$t\longmapsto\rho^{2t}$ and $t\longmapsto\sigma^{1-t}$
are real analytic. Hence,
$t\longmapsto\sigma^{1-t}\rho^{2t}\sigma^{1-t}$
is a real-analytic family of positive definite matrices. The principal
square-root map is real analytic on the positive definite matrices, so
$t\longmapsto\Tr\left(\sigma^{1-t}\rho^{2t}\sigma^{1-t}\right)^{1/2}=\Phi_t(\rho,\sigma)$
is positive and real analytic on $(0,1)$. It follows that
$g(t)=\log\Phi_t(\rho,\sigma)$
is also real analytic on $(0,1)$.
Since $g$ is affine on the nonempty interval $(a,b)$, we have
\[
g''(t)=0,
\qquad a<t<b.
\]
The function $g''$ is real analytic on $(0,1)$. By the identity
theorem for real-analytic functions, a real-analytic function that
vanishes on a nonempty open interval must vanish everywhere.
Consequently,
\[
g''(t)=0,
\qquad 0<t<1,
\]
and hence $g$ is affine on all of $(0,1)$.

It remains to determine this affine function. Since $\rho$ and
$\sigma$ are positive definite on their common support, functional
calculus gives
\[
\lim_{t\downarrow0}\rho^t=I
\qquad\text{and}\qquad
\lim_{t\uparrow1}\sigma^{1-t}=I.
\]
Therefore,
\[
\lim_{t\downarrow0}\Phi_t(\rho,\sigma)
=
\|\sigma\|_1
=
\Tr\sigma
=1,
\]
and similarly,
\[
\lim_{t\uparrow1}\Phi_t(\rho,\sigma)
=
\|\rho\|_1
=
\Tr\rho
=1.
\]
It follows that
\[
\lim_{t\downarrow0}g(t)
=
\lim_{t\uparrow1}g(t)
=0.
\]
An affine function whose limits at both endpoints are zero must be
identically zero. Hence,
\[
g(t)=0,
\qquad 0<t<1,
\]
or equivalently,
\[
\Phi_t(\rho,\sigma)=1,
\qquad 0<t<1.
\]
The upper-rigidity result, Theorem~\ref{thm:upper-rigidity}, now implies
that $\rho=\sigma$.
Conversely, if $\rho=\sigma$, then
\[
\Phi_t(\rho,\rho)
=\|\rho^t\rho^{1-t}\|_1
=\|\rho\|_1
=1
\]
for every $0<t<1$. In particular,
$F(\rho,\rho)^2=\Phi_\alpha(\rho,\rho)\Phi_{1-\alpha}(\rho,\rho)$,
which proves the converse.

Part \textnormal{(iii)} follows by combining \textnormal{(i)} and
\textnormal{(ii)}.  For \textnormal{(iv)}, let
$I=\{i:p_iq_i>0\}$.  If $I$ is nonempty, then
\[
 \Phi_t(\rho,\sigma)=\sum_{i\in I}q_i
 \exp\left(t\log\frac{p_i}{q_i}\right).
\]
The second derivative of its logarithm is the variance of
$\log(p_i/q_i)$ with respect to the corresponding tilted probability
weights.  It vanishes on a nontrivial interval exactly when these ratios
are constant on $I$.  This is precisely the stated condition.

If the ratios in \textnormal{(iv)} equal $c$, and
$Q=\sum_{i\in I}q_i$, then $\Phi_t(\rho,\sigma)=c^tQ$.
For $\alpha\ne1/2$, equality
$\Phi_\alpha=\Phi_{1-\alpha}$ is therefore equivalent to $c=1$.
Combining this observation with \textnormal{(i)} proves
\textnormal{(v)}.  Finally, orthogonal supports make every product
$\rho^t\sigma^{1-t}$ equal to zero.
\end{proof}

\begin{corollary}
\label{cor:functional-equality}
Let $\rho$ and $\sigma$ be finite-dimensional nonorthogonal states,
and fix $0<\alpha<1$ with $\alpha\neq\tfrac12$. Then the following
statements are equivalent:
\begin{enumerate}[label=\textnormal{(\roman*)}]
\item
$
F(\rho,\sigma)=\mathcal{F}_\alpha(\rho,\sigma);
$

\item equality holds in both of the first two inequalities in
Corollary~\ref{thm:midpoint};

\item the function
$
t\longmapsto\Phi_t(\rho,\sigma)
$
is constant on $(0,1)$.
\end{enumerate}
Whenever these equivalent conditions hold, the value of
$\Phi_t(\rho,\sigma)$ is independent of $t$, and for every $0<t<1$,
\[
F(\rho,\sigma)
=
\sqrt{\Phi_t(\rho,\sigma)\Phi_{1-t}(\rho,\sigma)}
=
\mathcal{F}_t(\rho,\sigma).
\]
This common value need not be equal to $1$.
\end{corollary}

\begin{proof}
Since $\rho$ and $\sigma$ are not orthogonal, their support subspaces
are not orthogonal. Positive powers do not change the support of a
positive operator. Hence, for every $0<t<1$,
\[
\rho^t\sigma^{1-t}\neq0,
\]
and therefore $\Phi_t(\rho,\sigma)=\|\rho^t\sigma^{1-t}\|_1>0$.
Thus the function
$g(t)=\log\Phi_t(\rho,\sigma)$
is well defined on $(0,1)$.
We first note that $g$ is real analytic. Indeed, by the spectral
theorem, the maps
\[
t\longmapsto\rho^t
\qquad\text{and}\qquad
t\longmapsto\sigma^{1-t}
\]
are real analytic for $0<t<1$, where the powers are taken to be zero
on the corresponding kernels. Moreover, positive powers have the same
support as the original operators. It follows that the rank of
$\rho^t\sigma^{1-t}$
is independent of $t\in(0,1)$. Consequently, the number of its nonzero
singular values is constant. These singular values stay separated
from zero in a sufficiently small neighborhood of each point
$t\in(0,1)$. Functional calculus applied to
\[
\bigl(\rho^t\sigma^{1-t}\bigr)^*
\bigl(\rho^t\sigma^{1-t}\bigr)
=
\sigma^{1-t}\rho^{2t}\sigma^{1-t}
\]
then shows that
\[
t\longmapsto
\Phi_t(\rho,\sigma)
=
\Tr\left(
\sigma^{1-t}\rho^{2t}\sigma^{1-t}
\right)^{1/2}
\]
is real analytic. Since this function is positive, $g$ is also real
analytic on $(0,1)$.

We now prove that \textnormal{(i)} implies \textnormal{(iii)}. Suppose
that
\[
F(\rho,\sigma)=\mathcal{F}_\alpha(\rho,\sigma).
\]
By Corollary~\ref{thm:midpoint},
\[
F(\rho,\sigma)
\leq
\sqrt{
\Phi_\alpha(\rho,\sigma)
\Phi_{1-\alpha}(\rho,\sigma)}
\leq
\mathcal{F}_\alpha(\rho,\sigma).
\]
Since the first and last terms are equal, equality must hold in both
inequalities. In particular,
\[
F(\rho,\sigma)^2
=
\Phi_\alpha(\rho,\sigma)
\Phi_{1-\alpha}(\rho,\sigma).
\]
Since $F(\rho,\sigma)=\Phi_{1/2}(\rho,\sigma)$,
taking logarithms gives
\[
2g\left(\frac12\right)
=
g(\alpha)+g(1-\alpha).
\]
By Theorem~\ref{thm:logconvex}, $g$ is convex. Thus the preceding
equality is equality in the convexity inequality at the midpoint of
$\alpha$ and $1-\alpha$. It follows that $g$ is affine on the interval
joining these two points.

Since $g$ is real analytic, this affine behavior extends to all of
$(0,1)$. More precisely, $g''$ vanishes on a nonempty open interval.
The function $g''$ is real analytic, so the identity theorem gives
\[
g''(t)=0,
\qquad 0<t<1.
\]
Therefore, $g$ is affine on $(0,1)$.

Equality in the second inequality of the midpoint chain is equality
in the arithmetic--geometric mean inequality. Hence,
$\Phi_\alpha(\rho,\sigma)=\Phi_{1-\alpha}(\rho,\sigma)$,
and therefore $g(\alpha)=g(1-\alpha)$.

An affine function that takes the same value at two distinct points must be constant. Since $\alpha\neq\tfrac12$, 
the points $\alpha$ and
$1-\alpha$ are distinct. Thus $g$ is constant on $(0,1)$, and so is
$t\longmapsto\Phi_t(\rho,\sigma)$.
This proves \textnormal{(i)}$\Rightarrow$\textnormal{(iii)}.

Suppose next that \textnormal{(iii)} holds. Then there exists a
constant $c>0$ such that
$\Phi_t(\rho,\sigma)=c$
for every $0<t<1$. In particular,
\[
F(\rho,\sigma)
=\Phi_{1/2}(\rho,\sigma)
=c,
\]
and
\[
\sqrt{
\Phi_\alpha(\rho,\sigma)
\Phi_{1-\alpha}(\rho,\sigma)}
=\sqrt{c^2}=c.
\]
Also,
\[
\mathcal{F}_\alpha(\rho,\sigma)
=\frac12(c+c)=c.
\]
Consequently,
\[
F(\rho,\sigma)
=
\sqrt{
\Phi_\alpha(\rho,\sigma)
\Phi_{1-\alpha}(\rho,\sigma)}
=
\mathcal{F}_\alpha(\rho,\sigma).
\]
Thus both of the first two inequalities in
Corollary~\ref{thm:midpoint} are equalities. This proves
\[
\textnormal{(iii)}
\Longrightarrow
\textnormal{(ii)}
\Longrightarrow
\textnormal{(i)}.
\]
Therefore, the three conditions are equivalent.
\end{proof}

\begin{theorem}
\label{thm:spectral-equality}
Let $\rho$ and $\sigma$ be nonorthogonal finite-dimensional states with
spectral decompositions over their distinct positive eigenvalues,
\[
 \rho=\sum_i r_iP_i,\quad \sigma=\sum_j s_jQ_j,
\]
and fix $\alpha\ne1/2$.  Equality in the log-convexity step of
Corollary~\ref{thm:midpoint} holds if and only if there is a constant $c>0$
such that
\begin{equation}\label{eq:spectral-edge-ratio}
 P_iQ_j\ne0\quad\Longrightarrow\quad r_i=cs_j.
\end{equation}
Moreover,
$ F(\rho,\sigma)=\mathcal F_\alpha(\rho,\sigma)$
if and only if
\begin{equation}\label{eq:spectral-edge-equality}
 P_iQ_j\ne0\quad\Longrightarrow\quad r_i=s_j.
\end{equation}
Together with the orthogonal case, this characterizes all equality pairs
in finite dimensions.
\end{theorem}

\begin{proof}
If \eqref{eq:spectral-edge-ratio} holds, then, for every real $t$,
\begin{align*}
 \rho^t\sigma^{1-t}
 &=\sum_{i,j}r_i^ts_j^{1-t}P_iQ_j\\
 &=c^t\sum_{i,j}s_jP_iQ_j
 =c^t s(\rho)\sigma.
\end{align*}
Consequently
\begin{equation}\label{eq:exponential-overlap-curve}
 \Phi_t(\rho,\sigma)=c^t\lVert s(\rho)\sigma\rVert_1.
\end{equation}
Thus \eqref{eq:spectral-edge-ratio} gives equality in the log-convexity
step.  Conversely, suppose equality holds in that step.  The argument in
Corollary~\ref{cor:functional-equality}, without using the
arithmetic--geometric equality, shows that
\begin{equation}\label{eq:affine-overlap-curve}
 \Phi_t(\rho,\sigma)=Kc^t
\end{equation}
on $(0,1)$ for some $K,c>0$.  The product
\[
 \rho^t\sigma^{1-t}
 =\sum_{P_iQ_j\ne0}s_j
   \left(\frac{r_i}{s_j}\right)^tP_iQ_j
\]
where negative powers are taken on the corresponding supports.  Both
factors are invertible on their supports, so the rank equals the rank of
$s(\rho)s(\sigma)$ and is independent of $t$.  The product therefore
depends real analytically on $t\in\mathbb R$ with constant rank.
Consequently its trace norm is real analytic, and
\eqref{eq:affine-overlap-curve} extends to every real $t$.

Let
\[
 M=\max_{P_iQ_j\ne0}\frac{r_i}{s_j},\qquad
 m=\min_{P_iQ_j\ne0}\frac{r_i}{s_j}.
\]
The matrices $P_iQ_j$ are mutually orthogonal in the Hilbert--Schmidt
inner product when the ordered pairs $(i,j)$ differ.  Hence the sum of
the terms corresponding to $M$ is nonzero, and comparison of the dominant
terms as $t\to+\infty$ gives
\[
 \lim_{t\to+\infty}\Phi_t^{1/t}=M.
\]
Equation \eqref{eq:affine-overlap-curve} gives the same limit as $c$, so
$M=c$.  Applying the same argument as $t\to-\infty$ gives $m=c$.
Therefore every active ratio equals $c$, proving
\eqref{eq:spectral-edge-ratio}.

This proves the first assertion.  For the second, the
arithmetic--geometric equality and
$\alpha\ne1/2$ force $c^\alpha=c^{1-\alpha}$, hence $c=1$.  This
classification contains the commuting characterization and the normalized
equal-rank projection examples, and also covers noncommuting states with
unequal overlapping supports.
\end{proof}

\section{Equality phenomena and rigidity}

The preceding results give a complete description of the midpoint
equality cases in finite dimensions. They also cover noncommuting
states whose supports overlap without being equal.

Equality does not require the two states to have the same support or to
commute. It may hold when the supports overlap but are different, and
also for noncommuting pairs. The following result gives a simple class
of such examples and illustrates Theorem~\ref{thm:spectral-equality}.

\begin{proposition}
Fix $\alpha\neq\tfrac12$. Equality
\[
 F(\rho,\sigma)=\mathcal F_\alpha(\rho,\sigma)
\]
holds if and only if
\[
 \Phi_\alpha(\rho,\sigma)=\Phi_{1-\alpha}(\rho,\sigma)
\]
and equality holds in the log-convexity estimate
\[
 \Phi_{1/2}(\rho,\sigma)^2
 =\Phi_\alpha(\rho,\sigma)
  \Phi_{1-\alpha}(\rho,\sigma).
\]
In particular, let $P$ and $Q$ be finite-rank projections of the same
rank $r$, and set $\rho=P/r$ and $\sigma=Q/r$. Then
\[
 \Phi_t(\rho,\sigma)=\frac1r\|PQ\|_1
 =F(\rho,\sigma),\qquad 0<t<1.
\]
Consequently the whole parameter curve is constant for this class,
whether or not $P$ and $Q$ commute.
\end{proposition}

\begin{proof}
Corollary~\ref{thm:midpoint} gives
\[
 F\leq\sqrt{\Phi_\alpha\Phi_{1-\alpha}}
 \leq\frac{\Phi_\alpha+\Phi_{1-\alpha}}2.
\]
Equality between the first and last terms forces equality at both steps.
Equality in the scalar arithmetic--geometric mean inequality is equivalent
to $\Phi_\alpha=\Phi_{1-\alpha}$, which proves the first assertion.

For the stated projection states, $P^t=P$ and $Q^t=Q$ for every $t>0$.
Thus
\[
 \rho^t\sigma^{1-t}=r^{-1}PQ,
\]
and the formula follows immediately.
\end{proof}

\begin{example}
In $\mathbb C^3$, let $P$ be the projection onto
$\operatorname{span}\{e_1,e_2\}$ and let $Q$ be the projection onto
\[
 \operatorname{span}\{e_1,
 \cos\vartheta\,e_2+\sin\vartheta\,e_3\},
 \qquad 0<\vartheta<\frac\pi2.
\]
Then $P$ and $Q$ do not commute, but for $\rho=P/2$ and $\sigma=Q/2$,
\[
 \mathcal F_t(\rho,\sigma)=F(\rho,\sigma)
 =\frac{1+\cos\vartheta}{2},\qquad 0<t<1.
\]
Thus midpoint equality does not characterize commutativity or equality of
the states.
\end{example}

\begin{theorem}
\label{thm:constancy-rigidity}
Let $\rho=\sum_i r_iP_i$ and $\sigma=\sum_j s_jQ_j$ be
finite-dimensional states, where only the distinct positive eigenvalues
are displayed.  The following are equivalent:
\begin{enumerate}[label=\textnormal{(\roman*)}]
 \item $t\mapsto\mathcal F_t(\rho,\sigma)$ is constant on a nontrivial
       open subinterval of $(0,1)$;
 \item it is constant on all of $(0,1)$;
 \item either $\rho\sigma=0$, or
 \[
  P_iQ_j\ne0\quad\Longrightarrow\quad r_i=s_j.
 \]
\end{enumerate}
In either case the constant value is $F(\rho,\sigma)$.
\end{theorem}

\begin{proof}
Only \textnormal{(i)} implies \textnormal{(iii)} requires proof.  The
orthogonal case is immediate, so suppose that the states are
nonorthogonal and put $f(t)=\Phi_t(\rho,\sigma)>0$.  By
Theorem~\ref{thm:logconvex}, both $f(t)$ and $f(1-t)$ are log-convex and
hence convex.  If
\[
 2\mathcal F_t=f(t)+f(1-t)
\]
is constant on an interval, the two nonnegative convexity gaps in this
sum must vanish separately.  Thus both summands are affine there.
Since $f$ is positive and affine, direct differentiation gives
\[
 (\log f)''=-\frac{(f')^2}{f^2}.
\]
Log-convexity makes the left-hand side nonnegative, and therefore
$f'=0$.  Hence $f$ is constant on a nontrivial interval.

Choose distinct $a,b$ in that interval and put $m=(a+b)/2$.  Then
$\Phi_m^2=\Phi_a\Phi_b$, so equality holds in log-convexity.  The proof
of Theorem~\ref{thm:spectral-equality}, applied to $[a,b]$ instead of the
symmetric pair $\{\alpha,1-\alpha\}$, yields $c>0$ such that
\[
 P_iQ_j\ne0\quad\Longrightarrow\quad r_i=cs_j.
\]
Equation~\eqref{eq:exponential-overlap-curve} gives $f(t)=Kc^t$.
Constancy on the original interval forces $c=1$, proving
\textnormal{(iii)}.  Conversely, \textnormal{(iii)} and
Equation~\eqref{eq:exponential-overlap-curve} with $c=1$ show that both
$\Phi_t$ and $\mathcal F_t$ are constant on $(0,1)$.  Evaluation at
$t=1/2$ identifies their value as $F$.
\end{proof}

\begin{theorem}
\label{thm:upper-rigidity}
For every $0<\alpha<1$, $\Phi_\alpha(\rho,\sigma)=1$ if and only if $\rho=\sigma$,
and hence $\mathcal F_\alpha(\rho,\sigma)=1$ if and only if $\rho=\sigma$.
\end{theorem}

\begin{proof}
Put $p=1/\alpha$, $q=1/(1-\alpha)$,
$a=\rho^\alpha$, and $b=\sigma^{1-\alpha}$. The trace version of
Young's inequality gives
\[
 \Tr|ab|\leq \frac1p\Tr(a^p)+\frac1q\Tr(b^q)=1.
\]
Its equality theorem, valid for positive compact operators without a
faithfulness assumption, states that equality holds precisely when
$a^p=b^q$; see \cite{FarenickManjegani,ManjeganiTrace}.
Here this condition is exactly $\rho=\sigma$. Conversely, if
$\rho=\sigma$, then
$\Phi_\alpha(\rho,\rho)=\Tr\rho=1$.

Finally, $\mathcal F_\alpha$ is the average of two numbers bounded above
by one. Its value is one exactly when both numbers equal one, and either
equality already implies $\rho=\sigma$.
\end{proof}

\section{The noncommutative defect}

We have so far studied the directed overlap
$\Phi_\alpha(\rho,\sigma)=\|\rho^\alpha\sigma^{1-\alpha}\|_1$.
A closely related quantity is the Petz overlap
$\Tr(\rho^\alpha\sigma^{1-\alpha})$,
which appears in the definition of the Petz R\'enyi divergence. The
two quantities agree when $\rho$ and $\sigma$ commute. In the
noncommutative case, however, the product
$\rho^\alpha\sigma^{1-\alpha}$ need not be positive, and its trace norm
may be strictly larger than its trace.
To measure this difference, we define
\[
\Delta_\alpha(\rho,\sigma)
:=
\|\rho^\alpha\sigma^{1-\alpha}\|_1
-\Tr(\rho^\alpha\sigma^{1-\alpha}),
\qquad 0<\alpha<1.
\]
This quantity is always nonnegative. Indeed,
\[
\Tr(\rho^\alpha\sigma^{1-\alpha})
=
\Tr\left(
\sigma^{(1-\alpha)/2}
\rho^\alpha
\sigma^{(1-\alpha)/2}
\right)\geq0,
\]
and the trace duality inequality gives
\[
\Tr(\rho^\alpha\sigma^{1-\alpha})
\leq
\|\rho^\alpha\sigma^{1-\alpha}\|_1.
\]
Therefore,
$\Delta_\alpha(\rho,\sigma)\geq0$.
If $\rho$ and $\sigma$ commute, then
$\rho^\alpha\sigma^{1-\alpha}$ is positive. Hence its trace norm is
equal to its trace, and consequently $\Delta_\alpha(\rho,\sigma)=0$.
The main result of this section proves the converse: the defect is zero
only when $\rho$ and $\sigma$ commute.

\begin{theorem}
\label{thm:defect-rigidity}
Let $0<\alpha<1$ and let $\rho,\sigma$ be density operators on a
separable Hilbert space.  Then
\[
 \Tr(\rho^\alpha\sigma^{1-\alpha})
 =\|\rho^\alpha\sigma^{1-\alpha}\|_1
\]
if and only if $\rho$ and $\sigma$ commute. Equivalently,
$ \Delta_\alpha(\rho,\sigma)=0$ if and only if 
$ \rho\sigma=\sigma\rho$.
No faithfulness or common-support hypothesis is required.
\end{theorem}

\begin{proof}
Put $X=\rho^\alpha\sigma^{1-\alpha}$.  Noncommutative H\H{o}lder
shows that $X$ is trace class.  Moreover,
\[
 \Tr X
 =\Tr\!\left(\sigma^{(1-\alpha)/2}
              \rho^\alpha
              \sigma^{(1-\alpha)/2}\right)\geq0,
\]
where the equality follows first for finite-rank spectral truncations and
then by trace-norm approximation.  For every trace-class operator $T$,
$|\Tr T|\leq\|T\|_1$, and equality with
$\Tr T=\|T\|_1\geq0$ holds precisely when $T$ is positive.  This is the
equality case in trace-norm duality (or follows immediately from the polar
decomposition of $T$).  Consequently, $\Tr X=\|X\|_1$ if and only if $X\geq0$.

The product of two bounded positive operators $A$ and $B$ is positive if
and only if it is self-adjoint, which is equivalent to $AB=BA$.  Applying
this with $A=\rho^\alpha$ and $B=\sigma^{1-\alpha}$ gives
\[
 X\geq0,
 \quad\text{if and only if} \quad
 \rho^\alpha\sigma^{1-\alpha}
 =\sigma^{1-\alpha}\rho^\alpha.
\]
Finally, the functions $t\mapsto t^\alpha$ and
$t\mapsto t^{1-\alpha}$ are continuous and injective on $[0,1]$;
their inverse functions are continuous on their respective spectra.
Functional calculus therefore shows that the last commutation relation is
equivalent to $\rho\sigma=\sigma\rho$.
\end{proof}

\begin{corollary}\label{cor:strict-petz-defect}
For every noncommuting pair of states and every $0<\alpha<1$,
\[
 \Tr(\rho^\alpha\sigma^{1-\alpha})
 <\|\rho^\alpha\sigma^{1-\alpha}\|_1.
\]
Thus $\Delta_\alpha$ detects noncommutativity. We do not claim that it is
a metric or a divergence.
\end{corollary}

Theorem~\ref{thm:defect-rigidity} explains the term noncommutative
defect and distinguishes the trace-norm overlap from the Petz and
Chernoff functionals.

\section{Channels and preserving maps}

A quantum channel is represented by a completely positive
trace-preserving map $\mathcal{E}$. For a quantity that compares two
quantum states, it is natural to ask how its value changes when the same
channel is applied to both states.

The standard fidelity is monotone under quantum channels. Thus, for
every pair of density operators $\rho$ and $\sigma$,
$F(\rho,\sigma)\leq F\bigl(\mathcal{E}(\rho),\mathcal{E}(\sigma)\bigr)$.
Thus a quantum channel cannot decrease the fidelity between two
states. It is natural to ask whether the same property holds for the
symmetrized overlap $\mathcal{F}_\alpha$.

Before studying this question, we recall what is already known for the
directed overlap $\Phi_\alpha$. By
Proposition~\ref{prop:alpha-z},
\[
\Phi_\alpha(\rho,\sigma)
=
Q_{\alpha,1/2}(\rho\Vert\sigma),
\]
so its behavior under quantum channels is governed by the
data-processing theory of the $\alpha$-$z$ R\'enyi quantities. For
$0<\beta<1$, the exact finite-dimensional data-processing region for
$D_{\beta,z}$ is
\[
z\geq\max\{\beta,1-\beta\};
\]
see \cite{HiaiJencova,Zhang}. If $z=\tfrac12$, this condition is
satisfied only when
$\beta=\frac12$.
Consequently, the established $\alpha$-$z$ theory shows that the
directed overlap $\Phi_\alpha$ is monotone under all quantum channels
only at the midpoint $\alpha=\tfrac12$.

This result does not answer the same question for the
symmetrized quantity
\[
\mathcal{F}_\alpha(\rho,\sigma)
=
\frac12\left(
\Phi_\alpha(\rho,\sigma)
+\Phi_{1-\alpha}(\rho,\sigma)
\right).
\]
Neither directed term is monotone for all channels away from the
midpoint. However, their arithmetic mean might still be monotone. In
this section, we show that averaging does not restore monotonicity.

For this purpose, we extend
$\mathcal{F}_\alpha$ from density operators to arbitrary positive
semidefinite matrices. For $A,B\geq0$, define
\[
G_\alpha(A,B)
:=
\frac12\left(
\|A^\alpha B^{1-\alpha}\|_1
+
\|A^{1-\alpha}B^\alpha\|_1
\right).
\]
When $A$ and $B$ are density operators,
\[
G_\alpha(A,B)=\mathcal{F}_\alpha(A,B).
\]
The function $G_\alpha$ is positively homogeneous, since
\[
G_\alpha(cA,cB)=c\,G_\alpha(A,B)
\]
for every $c\geq0$. This extension relates data processing to joint
concavity.

\begin{theorem}
\label{thm:dpi-concavity}
Fix $0<\alpha<1$. The following are equivalent in finite dimensions:
\begin{enumerate}[label=\textup{(\roman*)}]
\item $G_\alpha$ is jointly concave on the equal-trace cone
      \[
       \{(A,B):A,B\geq0,\ \Tr A=\Tr B\};
      \]
\item for every quantum channel $\mathcal E$ and all density matrices
      $\rho,\sigma$,
      \[
       \mathcal F_\alpha(\rho,\sigma)
       \leq\mathcal F_\alpha(\mathcal E(\rho),
       \mathcal E(\sigma)).
      \]
\end{enumerate}
\end{theorem}

\begin{proof}
Assume first that $G_\alpha$ is jointly concave. It is invariant under
simultaneous unitary conjugation and additive under matched direct sums.
Moreover, for the maximally mixed state $\omega_m=I_m/m$,
\[
 G_\alpha(A\otimes\omega_m,B\otimes\omega_m)
 =G_\alpha(A,B),
\]
because $\Phi_t(\omega_m,\omega_m)=1$ and powers and trace norms are
multiplicative under tensor products.

Let $A,B$ act on $H\otimes\mathbb C^m$. Averaging simultaneous conjugation
by a unitary error basis $\{U_j\}_{j=1}^{m^2}$ on the second tensor factor
gives
\[
 \frac1{m^2}\sum_{j=1}^{m^2}
 (I\otimes U_j)A(I\otimes U_j)^*
 =\Tr_{\mathbb C^m}(A)\otimes\frac{I_m}{m},
\]
and the analogous identity for $B$. Joint concavity and unitary invariance
therefore imply
\begin{align*}
 G_\alpha\left(\Tr_{\mathbb C^m}A,
                    \Tr_{\mathbb C^m}B\right)
 &=G_\alpha\left(
   \Tr_{\mathbb C^m}(A)\otimes\frac{I_m}{m},
   \Tr_{\mathbb C^m}(B)\otimes\frac{I_m}{m}\right)\\
 &\geq G_\alpha(A,B).
\end{align*}
Every quantum channel has a Stinespring representation as an isometry
followed by a partial trace. Since $G_\alpha$ is invariant under a common
isometry, (ii) follows.

Conversely, suppose (ii) holds. Let $A_1,A_2,B_1,B_2$ be density matrices
and $0\leq\lambda\leq1$. Introduce a two-dimensional classical register
and the block-diagonal states
\begin{align*}
 \widehat A&=\lambda A_1\otimes|0\rangle\langle0|
 +(1-\lambda)A_2\otimes|1\rangle\langle1|,\\
 \widehat B&=\lambda B_1\otimes|0\rangle\langle0|
 +(1-\lambda)B_2\otimes|1\rangle\langle1|.
\end{align*}
Block additivity and homogeneity give
\[
 G_\alpha(\widehat A,\widehat B)
 =\lambda G_\alpha(A_1,B_1)
 +(1-\lambda)G_\alpha(A_2,B_2).
\]
Applying (ii) to the partial trace over the classical register yields
joint concavity on states. Positive homogeneity then extends it to all
positive pairs with equal trace, which is exactly the form needed for
channel monotonicity.
\end{proof}

Although data processing fails for some mixed inputs, as shown below, it
holds on a noncommutative class of inputs that includes every pair of pure
states.

\begin{proposition}
\label{prop:dpi-projections}
Let $P,Q$ be finite-rank projections of the same rank $r$, and put
$\rho=P/r$ and $\sigma=Q/r$.  For every $0<\alpha<1$ and every quantum
channel $\mathcal E$,
\[
 \mathcal F_\alpha(\rho,\sigma)
 \leq
 \mathcal F_\alpha\bigl(\mathcal E(\rho),\mathcal E(\sigma)\bigr).
\]
In particular, this holds for every pair of pure input states.
\end{proposition}

\begin{proof}
Functional calculus on the supports gives
\[
 \rho^\alpha\sigma^{1-\alpha}=\frac1rPQ,
 \qquad
 \rho^{1-\alpha}\sigma^\alpha=\frac1rPQ.
\]
Hence
\[
 \mathcal F_\alpha(\rho,\sigma)
 =\frac1r\lVert PQ\rVert_1
 =F(\rho,\sigma),
\]
where $F$ denotes root fidelity.  The midpoint bound proved above and
monotonicity of root fidelity under quantum channels now yield
\[
 \mathcal F_\alpha(\rho,\sigma)=F(\rho,\sigma)
 \leq F\bigl(\mathcal E(\rho),\mathcal E(\sigma)\bigr)
 \leq \mathcal F_\alpha
      \bigl(\mathcal E(\rho),\mathcal E(\sigma)\bigr).
\]
\end{proof}

We now determine the values of $\alpha$ for which $G_\alpha$ is jointly
concave on the equal-trace cone. By the preceding theorem, this is
equivalent to asking when
\[
 \cF_\alpha(\rho,\sigma)\leq
 \cF_\alpha(\mathcal E(\rho),\mathcal E(\sigma))
\]
holds for all states and all quantum channels.
This question has a useful one-variable form. On
$H\oplus H$ put
\[
 R=\begin{pmatrix}A&0\\0&B\end{pmatrix},\qquad
 J=\begin{pmatrix}0&I\\I&0\end{pmatrix}.
\]
Then $J=J^*=J^{-1}$ and direct block multiplication gives
\begin{equation}\label{eq:block-swap-reduction}
 R^\alpha J R^{1-\alpha}
 =\begin{pmatrix}
 0&A^\alpha B^{1-\alpha}\\
 B^\alpha A^{1-\alpha}&0
 \end{pmatrix}.
\end{equation}
Since the trace norm of an off-diagonal block matrix
$\left(\begin{smallmatrix}0&X\\Y&0\end{smallmatrix}\right)$ equals
$\lVert X\rVert_1+\lVert Y\rVert_1$, we obtain the exact identity
\begin{equation}\label{eq:G-block-swap}
 2G_\alpha(A,B)=\lVert R^\alpha J R^{1-\alpha}\rVert_1.
\end{equation}
Thus, joint concavity of $G_\alpha$ is equivalent to concavity of the
right-hand side of \eqref{eq:G-block-swap} on the block-diagonal
positive cone. In terms of the $\alpha$--$z$ trace functions, the same
identity is
\[
 2G_\alpha(A,B)=Q_{\alpha,1/2}(R\Vert JRJ).
\]
This is a special case of the two-variable $\alpha$--$z$ problem,
because the second argument is a fixed unitary conjugate of the first.
The general concavity theorem for $Q_{\alpha,z}$ does not decide this
special case away from the midpoint. Indeed, when $z=1/2$, its
concavity range requires
$1/2\geq\max\{\alpha,1-\alpha\}$, and hence $\alpha=1/2$. Therefore,
other values of $\alpha$ must be studied by using the block-swap form
in \eqref{eq:block-swap-reduction}.

For qubits, this problem can be written in scalar form.

\begin{proposition}
\label{prop:qubit-formula}
Let $A,B\in M_2^+$ and $0<\alpha<1$.  Then
\begin{equation}\label{eq:qubit-directed}
 \Phi_\alpha(A,B)^2
 =\Tr\!\left(A^{2\alpha}B^{2(1-\alpha)}\right)
  +2(\det A)^\alpha(\det B)^{1-\alpha}.
\end{equation}
Consequently,
\begin{align}\label{eq:qubit-symmetrized}
 2G_\alpha(A,B)
 &=\left[
 \Tr\!\left(A^{2\alpha}B^{2(1-\alpha)}\right)
 +2(\det A)^\alpha(\det B)^{1-\alpha}
 \right]^{1/2}\notag\\
 &\quad+\left[
 \Tr\!\left(A^{2(1-\alpha)}B^{2\alpha}\right)
 +2(\det A)^{1-\alpha}(\det B)^\alpha
 \right]^{1/2}.
\end{align}
\end{proposition}

\begin{proof}
For every $2\times2$ matrix $X$, its two singular values give
\[
 \lVert X\rVert_1^2=\Tr(X^*X)+2|\det X|.
\]
Apply this identity to $X=A^\alpha B^{1-\alpha}$.  Cyclicity of the
trace yields
\[
 \Tr(X^*X)
 =\Tr\!\left(B^{1-\alpha}A^{2\alpha}B^{1-\alpha}\right)
 =\Tr\!\left(A^{2\alpha}B^{2(1-\alpha)}\right),
\]
whereas positivity gives
$|\det X|=(\det A)^\alpha(\det B)^{1-\alpha}$.  This proves
\eqref{eq:qubit-directed}; applying it at $\alpha$ and $1-\alpha$ gives
\eqref{eq:qubit-symmetrized}.
\end{proof}

For density matrices one may make \eqref{eq:qubit-symmetrized} entirely
scalar.  Write
\[
 A=\frac12(I+r\,u\cdot\boldsymbol\sigma),\qquad
 B=\frac12(I+s\,v\cdot\boldsymbol\sigma),
\]
where $0\leq r,s\leq1$ and $u,v$ are unit vectors.  If
\[
 c_p(t)=\frac12\left[\left(\frac{1+t}{2}\right)^p
                         +\left(\frac{1-t}{2}\right)^p\right],\qquad
 d_p(t)=\frac12\left[\left(\frac{1+t}{2}\right)^p
                         -\left(\frac{1-t}{2}\right)^p\right],
\]
then
\begin{equation}\label{eq:qubit-bloch-trace}
 \Tr(A^pB^q)=2\{c_p(r)c_q(s)+d_p(r)d_q(s)\,u\cdot v\},
 \qquad
 \det A=\frac{1-r^2}{4}.
\end{equation}
Equations \eqref{eq:qubit-symmetrized}--\eqref{eq:qubit-bloch-trace}
reduce qubit joint concavity to an explicit inequality in
$r,s,u\cdot v$.  They also show why the conclusion is not immediate from
Lieb concavity: outside the midpoint, one of the exponents $2\alpha$ and
$2(1-\alpha)$ exceeds one.  Thus the two square-root terms must be treated
together if symmetrization restores concavity.

There is also a dimension-free local test at the maximally mixed state.

\begin{proposition}
\label{prop:maximally-mixed-expansion}
Let $X,Y\in M_n$ be traceless Hermitian matrices and, for sufficiently
small real $\varepsilon$, set
\[
 \rho_\varepsilon=\frac{I+\varepsilon X}{n},\qquad
 \sigma_\varepsilon=\frac{I+\varepsilon Y}{n}.
\]
Then, for every $0<\alpha<1$,
\begin{equation}\label{eq:maximally-mixed-expansion}
 \mathcal F_\alpha(\rho_\varepsilon,\sigma_\varepsilon)
 =1-\frac{\alpha(1-\alpha)}{2n}\varepsilon^2
       \Tr (X-Y)^2+O(\varepsilon^3).
\end{equation}
In particular, the joint Hessian of $\mathcal F_\alpha$ at
$(I/n,I/n)$ is negative semidefinite; its null directions are precisely
the common perturbations $X=Y$.
\end{proposition}

\begin{proof}
The power-series expansion gives
\begin{align*}
 (I+\varepsilon X)^\alpha(I+\varepsilon Y)^{1-\alpha}
 &=I+\varepsilon K+\varepsilon^2L+O(\varepsilon^3),\\
 K&=\alpha X+(1-\alpha)Y,\\
 L&=-\frac{\alpha(1-\alpha)}2(X^2+Y^2)
       +\alpha(1-\alpha)XY.
\end{align*}
Here $K$ is Hermitian.  If
$Z_\varepsilon=I+\varepsilon K+\varepsilon^2L+O(\varepsilon^3)$, then
expanding $(Z_\varepsilon^*Z_\varepsilon)^{1/2}$ at the identity yields
\[
 \lVert Z_\varepsilon\rVert_1
 =n+\varepsilon\Tr K+\varepsilon^2\Re\Tr L
   +O(\varepsilon^3).
\]
Since $X$ and $Y$ are traceless,
\[
 \Tr K=0,\qquad
 \Re\Tr L=-\frac{\alpha(1-\alpha)}2\Tr(X-Y)^2.
\]
After inserting the factor $1/n$, this proves the asserted expansion
for $\Phi_\alpha$.  Replacing $\alpha$ by $1-\alpha$ leaves its quadratic
coefficient unchanged, and hence proves \eqref{eq:maximally-mixed-expansion}.
\end{proof}

\begin{corollary}
Let $\mathcal E:M_n\to M_n$ be completely positive, trace preserving,
and unital.  With the notation of Proposition
\ref{prop:maximally-mixed-expansion},
\[
 \mathcal F_\alpha(\mathcal E(\rho_\varepsilon),
                    \mathcal E(\sigma_\varepsilon))
 -\mathcal F_\alpha(\rho_\varepsilon,\sigma_\varepsilon)
 =\frac{\alpha(1-\alpha)}{2n}\varepsilon^2
 \left\{\Tr(X-Y)^2-\Tr\mathcal E(X-Y)^2\right\}
 +O(\varepsilon^3),
\]
and the displayed quadratic coefficient is nonnegative.  Equivalently,
data processing holds to second order at the maximally mixed state.
\end{corollary}

\begin{proof}
Unitality fixes $I/n$.  By the Kadison--Schwarz inequality and trace
preservation,
\[
 \Tr\mathcal E(H)^2\leq\Tr\mathcal E(H^2)=\Tr H^2
 \]
for every Hermitian $H$.  Apply
\eqref{eq:maximally-mixed-expansion} with $H=X-Y$.
\end{proof}

The local conclusion cannot be extended globally.  In fact, the midpoint
is the unique universally monotone parameter.

\begin{theorem}
\label{thm:dpi-counterexample}
For $0<\alpha<1$, the following are equivalent:
\begin{enumerate}[label=\textnormal{(\roman*)}]
 \item $\alpha=1/2$;
 \item $G_\alpha$ is jointly concave on density matrices of every size;
 \item for all density matrices and all quantum channels $\mathcal E$,
 \[
  \mathcal F_\alpha(\rho,\sigma)
  \leq\mathcal F_\alpha(\mathcal E(\rho),\mathcal E(\sigma)).
 \]
\end{enumerate}
If $\alpha\ne1/2$, failure in \textnormal{(ii)} and \textnormal{(iii)}
already occurs for diagonal pinching of real qubit states.  The two input
states may be chosen faithful.
\end{theorem}

\begin{proof}
The equivalence of \textnormal{(ii)} and \textnormal{(iii)} is
Theorem~\ref{thm:dpi-concavity}.  At $\alpha=1/2$ the functional is root
fidelity, so its joint concavity and data processing are standard.

It remains to prove necessity.  Since $G_\alpha=G_{1-\alpha}$, it is enough
to consider $0<\alpha<1/2$.  Put $d=1-b$, where $1/2<b<1$, and define
\[
 P=\begin{pmatrix}0&0\\0&1\end{pmatrix},\qquad
 B_y=\begin{pmatrix}b&y\\y&d\end{pmatrix},
 \qquad |y|<\sqrt{bd}.
\]
Let $\mathcal P$ denote diagonal pinching, so that
$\mathcal P(P)=P$ and $\mathcal P(B_y)=B_0$.  For $p>0$ set
\[
 h_p(b,d)=\frac{b^p-d^p}{b-d},\qquad
 k_p(b,d)=\frac{h_p(b,d)-p d^{p-1}}{b-d}.
\]
The second-order divided-difference formula for matrix powers gives
\begin{equation}\label{eq:power-22-expansion}
 (B_y^p)_{22}=d^p+k_p(b,d)y^2+O(y^4).
\end{equation}
Indeed, the two eigenvalues are
$b+y^2/(b-d)+O(y^4)$ and $d-y^2/(b-d)+O(y^4)$; expanding the corresponding
spectral projections gives precisely the coefficient $k_p(b,d)$ above.
Because $P^t=P$ for every $t>0$, the rank-one formula yields
\[
 \Phi_\alpha(P,B_y)
 =\sqrt{(B_y^{2(1-\alpha)})_{22}},\qquad
 \Phi_{1-\alpha}(P,B_y)
 =\sqrt{(B_y^{2\alpha})_{22}}.
\]
Using \eqref{eq:power-22-expansion}, we therefore obtain
\begin{equation}\label{eq:pinching-transverse-expansion}
 \mathcal F_\alpha(P,B_y)
 =\mathcal F_\alpha(P,B_0)+C_\alpha(d)y^2+O(y^4),
\end{equation}
where, with $b=1-d$,
\[
 C_\alpha(d)=\frac14\left\lbrace
 \frac{k_{2(1-\alpha)}(b,d)}{d^{1-\alpha}}
 +\frac{k_{2\alpha}(b,d)}{d^\alpha}
 \right\rbrace.
\]
As $d\downarrow0$, one has $b-d\to1$ and
$h_p(b,d)\to1$ for every $p>0$.  Hence
\begin{align*}
 \frac{k_{2(1-\alpha)}(b,d)}{d^{1-\alpha}}
 &=d^{-(1-\alpha)}(1+o(1)),\\
 \frac{k_{2\alpha}(b,d)}{d^\alpha}
 &=-2\alpha d^{-(1-\alpha)}(1+o(1)),
\end{align*}
and consequently
\begin{equation}\label{eq:positive-transverse-coefficient}
 C_\alpha(d)
 =\frac{1-2\alpha}{4}d^{-(1-\alpha)}(1+o(1))>0
\end{equation}
for all sufficiently small $d>0$.  Fix such a $d$.  Equation
\eqref{eq:pinching-transverse-expansion} then gives, for every sufficiently
small nonzero $y$,
\[
 \mathcal F_\alpha(P,B_y)
 >\mathcal F_\alpha(P,B_0)
 =\mathcal F_\alpha(\mathcal P(P),\mathcal P(B_y)).
\]
Thus diagonal pinching violates data processing.  Finally replace $P$ by
$P_\varepsilon=\operatorname{diag}(\varepsilon,1-\varepsilon)$.  The
inequality is strict and matrix powers and the trace norm are continuous
on the positive cone, so it persists for all sufficiently small
$\varepsilon>0$.  Both $P_\varepsilon$ and $B_y$ are then faithful real
qubit states.
\end{proof}

Here is a numerical example with faithful states. Let $\alpha=0.3$ and
\[
 \rho=\begin{pmatrix}
 0.196349286&0.336392366\\
 0.336392366&0.803650714
 \end{pmatrix},
 \qquad
 \sigma=\begin{pmatrix}
 0.999998445&0.000136430\\
 0.000136430&0.000001555
 \end{pmatrix}.
\]
Indeed,
\begin{align*}
 \mathcal F_{0.3}(\rho,\sigma)
 &=0.477142665205840\ldots,\\
 \mathcal F_{0.3}(\mathcal P(\rho),\mathcal P(\sigma))
 &=0.474608669456572\ldots,
\end{align*}
so the decrease under pinching is
$0.002533995749268\ldots$.  The smallest eigenvalues of $\rho$ and
$\sigma$ are approximately $0.04682941$ and $1.5364\times10^{-6}$,
respectively, so both states are faithful. Thus symmetrization does not
restore data processing outside the midpoint. On the other hand, data
processing holds for the special input class in
Proposition~\ref{prop:dpi-projections} and to second order at the
maximally mixed state.

\section{Concluding remarks and further directions}

We end with two possible directions for further study.

The first concerns maps that preserve the symmetrized overlap. Let
$\mathcal{E}$ be a surjective positive trace-preserving map satisfying
\[
\mathcal{F}_\alpha
\bigl(\mathcal{E}(\rho),\mathcal{E}(\sigma)\bigr)
=
\mathcal{F}_\alpha(\rho,\sigma)
\]
for all states $\rho$ and $\sigma$. Since
\[
\mathcal{F}_\alpha(\rho,\sigma)=0
\quad\Longleftrightarrow\quad
\rho\sigma=0,
\]
such a map preserves orthogonality in both directions. This suggests
using known results on orthogonality-preserving maps. In particular,
the methods developed by
Farenick, Jaques, and Rahaman~\cite{Farenick} for fidelity-preserving
maps may be useful here. One may ask for the form of these maps and the
weakest assumptions that give such a description.

The second direction concerns semifinite von Neumann algebras. Let
$(M,\tau)$ be a semifinite von Neumann algebra with a faithful normal
semifinite trace, and define
\[
\mathcal{S}_\tau
=
\{\rho\in L^1(M,\tau)_+:\tau(\rho)=1\}.
\]
For $\rho,\sigma\in\mathcal{S}_\tau$ and $0<\alpha<1$, we may consider
\[
\Phi_\alpha^\tau(\rho,\sigma)
=
\tau\left(
\left|\rho^\alpha\sigma^{1-\alpha}\right|
\right)
\]
and its symmetrized form
\[
\mathcal{F}_\alpha^\tau(\rho,\sigma)
=
\frac12\left(
\Phi_\alpha^\tau(\rho,\sigma)
+
\Phi_{1-\alpha}^\tau(\rho,\sigma)
\right).
\]
The noncommutative H\"older inequality gives
\[
0\leq\Phi_\alpha^\tau(\rho,\sigma)\leq1
\]
and hence
\[
0\leq\mathcal{F}_\alpha^\tau(\rho,\sigma)\leq1.
\]

The directed quantity $\Phi_\alpha^\tau$ is already part of the
von Neumann algebraic theory of the $\alpha$-$z$ R\'enyi divergences
developed through Haagerup $L^p$-spaces; see
\cite{HiaiJencova}. Thus, defining $\Phi_\alpha^\tau$ alone does not
give a new extension. The remaining questions concern the symmetrized
quantity. For example, one may ask whether our equality conditions
remain valid, how to state them for continuous spectra, and which
normal positive maps preserve $\mathcal{F}_\alpha^\tau$.

The infinite-trace case also requires some care. If $\tau(I)=\infty$,
the usual regularization
$x+\varepsilon I$
need not belong to $L^1(M,\tau)$ and cannot be used directly. Instead,
one must work on the supports of the operators or use finite-trace
corners. We leave these questions for future work.

\section{Conclusion}

In this paper, we studied a family of trace-norm overlaps between
quantum states. The midpoint
$\alpha=\tfrac12$ has a special role in this family: it gives the
standard root fidelity and is the only parameter for which the
symmetrized overlap is monotone under all quantum channels.

We proved the log-convexity of the directed overlap and used it to
obtain sharp bounds in terms of the standard fidelity. We also
described the equality cases in finite dimensions and determined
exactly when the overlap remains constant as the parameter varies.
These results show that equality may occur not only for commuting
states, but also for certain genuinely noncommuting pairs.

We further compared the directed trace-norm overlap with the Petz
overlap. The difference between these two quantities is always
nonnegative, and it vanishes exactly when the two states commute.
Finally, we showed that universal data processing for the symmetrized
overlap fails at every parameter other than $\alpha=\tfrac12$. This
failure already occurs for faithful real qubit states under a diagonal
pinching channel.

These results explain the role of the midpoint, describe the equality
cases, and show how noncommutativity affects this family of overlaps.
They also connect trace-norm interpolation with quantum fidelity and
the $\alpha$-$z$ R\'enyi theory.

\end{document}